\documentclass[10pt,journal,compsoc]{IEEEtran}

\ifCLASSOPTIONcompsoc
  \usepackage[nocompress]{cite}
\else
  \usepackage{cite}
\fi

\usepackage[pdftex]{graphicx}
\DeclareGraphicsExtensions{.pdf,.jpeg,.png}
\usepackage{amsmath}
\usepackage{amsfonts}
\usepackage{amssymb}
\usepackage{algorithm}
\usepackage[noend]{algpseudocode}
\usepackage{array}
\usepackage{multirow}

\usepackage[font=small,labelfont=bf]{caption}
\usepackage[font=small]{subcaption}
\usepackage{url}
\usepackage{times}
\usepackage{datetime}
\usepackage{hyperref}
\usepackage{flushend} 
\usepackage{amsthm}

\newtheorem{lemma}{Lemma}
\theoremstyle{definition}
\newtheorem{definition}{Definition}

\usepackage{color}
\usepackage{xcolor}
\usepackage{listings}
\usepackage{bbding}

\begin{document}

\title{WPaxos: Wide Area Network Flexible Consensus}

\author{
\IEEEauthorblockN{Ailidani Ailijiang,
Aleksey Charapko,
Murat Demirbas and
Tevfik Kosar}\\
\IEEEauthorblockA{Department of Computer Science and Engineering\\
University at Buffalo, SUNY\\
Email: \{ailidani,acharapk,demirbas,tkosar\}@buffalo.edu
}}

\IEEEtitleabstractindextext{%
\begin{abstract}
WPaxos is a multileader Paxos protocol that provides low-latency and high-throughput consensus across wide-area network (WAN) deployments. WPaxos uses multileaders, and partitions the object-space among these multileaders. Unlike statically partitioned multiple Paxos deployments, WPaxos is able to adapt to the changing access locality through object stealing. Multiple concurrent leaders coinciding in different zones steal ownership of objects from each other using phase-1 of Paxos, and then use phase-2 to commit update-requests on these objects locally until they are stolen by other leaders. To achieve fast phase-2 commits, WPaxos adopts the flexible quorums idea in a novel manner, and appoints phase-2 acceptors to be close to their respective leaders. We implemented WPaxos and evaluated it on WAN deployments across 5 AWS regions. The dynamic partitioning of the object-space and emphasis on zone-local commits allow WPaxos to significantly outperform both partitioned Paxos deployments and leaderless Paxos approaches.

\end{abstract}

}

\maketitle
\IEEEdisplaynontitleabstractindextext
\IEEEpeerreviewmaketitle

\IEEEraisesectionheading{\section{Introduction}\label{sec:intro}}


\IEEEPARstart{P}{axos} \cite{paxos} provides a formally-proven solution to the fault-tolerant distributed consensus problem. Notably, Paxos never violates the safety specification of distributed consensus (i.e., no two nodes decide differently), even in the case of fully asynchronous execution, crash/recovery of the nodes, and arbitrary loss of messages. When the conditions improve such that distributed consensus becomes solvable~\cite{flp,lamport1982generals}, Paxos also satisfies the progress property (i.e., nodes decide on a value as a function of the inputs). 
%
Paxos and its variants have been deployed widely, including in Chubby~\cite{chubby} based on Paxos~\cite{paxosmadecomplex}, Apache ZooKeeper~\cite{zookeeper} based on Zab~\cite{zab}, and etcd~\cite{etcd} based on Raft~\cite{raft}. 
These Paxos implementations depend on a centralized primary process (i.e., the leader) to serialize all commands. Due to this dependence on a single centralized leader, these Paxos implementations support deployments in local area and cannot deal with write-intensive scenarios across wide-area networks (WANs) well.
In recent years, however, coordination over wide-area networks (e.g., across zones, such as datacenters and sites) has gained greater importance, especially for database applications and NewSQL datastores~\cite{spanner,megastore,xie2014salt}, distributed filesystems~\cite{quintero2011, Mashtizadeh_2013, Grimshaw_2013}, and social networks~\cite{tao,cops}.

In order to eliminate the single leader bottleneck in Paxos, leaderless and multileader solutions were proposed. EPaxos \cite{epaxos} is a leaderless extension of the Paxos protocol where any replica at any zone can propose and commit commands opportunistically, provided that the commands are non-interfering. This opportunistic commit protocol requires an agreement from a fast-quorum of roughly 3/4th of the acceptors\footnote{For a deployment of size $2F+1$, fast-quorum is $F+\lfloor \frac{F+1}{2} \rfloor$}, which means that WAN latencies are still incurred. Moreover, if the commands proposed by multiple concurrent opportunistic proposers do interfere, the protocol requires performing a second phase to record the acquired dependencies, and agreement from a majority of the Paxos acceptors is needed.
Another way to eliminate the single leader bottleneck is to use a separate Paxos group deployed at each zone. Systems like Google Spanner \cite{spanner}, ZooNet \cite{zoonet}, and Bizur \cite{bizur} achieve this via a static partitioning of the global object-space to different zones, each responsible for a shard of the object-space. However, such static partitioning is inflexible and WAN latencies will be incurred persistently to access/update an object mapped to a different zone.

{\bf Contributions.}
We present {\em WPaxos}, a novel multileader Paxos protocol that provides low-latency and high-throughput consensus across WAN deployments.
WPaxos leverages the {\em flexible quorums}~\cite{fpaxos} idea to cut WAN communication costs. It deploys flexible quorums in a novel manner to appoint  {\em multiple concurrent leaders} across the WAN.
Unlike the FPaxos protocol~\cite{fpaxos} which uses a single-leader and does not scale to WAN distances, WPaxos uses multileaders and partitions the object-space among these multileaders. This allows the protocol to process requests for objects under different leaders concurrently. Each object in the system is maintained in its own commit log, allowing for per-object linearizability.
By strategically selecting the phase-2 acceptors to be close to the leader, WPaxos achieves fast commit decisions.
On the other hand, WPaxos differs from the existing static partitioned multiple Paxos deployment solutions because it implements a dynamic partitioning scheme: leaders coinciding in different zones steal ownership/leadership of an object from each other using phase-1 of Paxos, and then use phase-2 to commit update-requests on the object locally until the object is stolen by another leader.

With its multileader protocol, WPaxos guarantees linearizability per object.
We model WPaxos in TLA+/PlusCal~\cite{tla}
and present the algorithm using the PlusCal specification in Section~\ref{sec:algorithm}. The consistency properties of WPaxos are verified by model checking this specification\footnote{The TLA+ specification of WPaxos is available at \url{http://github.com/ailidani/paxi/tree/master/tla}}.

Since object stealing is an integrated part of phase-1 of Paxos, WPaxos remains simple as a pure Paxos flavor and obviates the need for another service/protocol. There is no need for a configuration service for relocating objects to zones as in Spanner~\cite{spanner} and vertical Paxos~\cite{verticalpaxos}.
Since the base WPaxos protocol guarantees safety to concurrency, asynchrony, and faults, the performance can be tuned orthogonally and aggressively, as we discuss in Section~\ref{sec:extensions}. To improve performance, we present a locality adaptive object stealing extension in Section~\ref{sec:adaptive}. 

To quantify the performance benefits from WPaxos, we implemented WPaxos in Go\footnote{The GO implementation of WPaxos is available at \url{http://github.com/ailidani/paxi}} and performed evaluations on WAN deployments across 5 AWS regions. Our results in Section~\ref{sec:eval} show that WPaxos outperforms EPaxos, achieving 15 times faster average request latency than EPaxos using a $\sim$70\% access locality workload in some regions. Moreover, for a $\sim$90\% access locality workload, WPaxos improves further and achieves 39 times faster average request latency than EPaxos in some regions. This is because, while the EPaxos opportunistic commit protocol requires about 3/4th of the Paxos acceptors to agree and incurs one WAN round-trip latency, WPaxos is able to achieve low latency commits using the zone-local phase-2 acceptors.
Moreover, WPaxos is able to maintain low-latency responses under a heavy workload: Under 10,000 requests/sec, using a $\sim$70\% access locality workload, WPaxos achieves 9 times faster average request latency and 54 times faster median latency than EPaxos.
Finally, we evaluate WPaxos with a shifting locality workload and show that WPaxos seamlessly adapts and significantly outperforms static partitioned multiple Paxos deployments.


While achieving low latency and high throughput, WPaxos also achieves seamless high-availability by having multileaders: failure of a leader is handled gracefully as other leaders can serve the requests previously processed by that leader via the object stealing mechanism. Since leader re-election (i.e., object stealing) is handled through the Paxos protocol, safety is always upheld to the face of node failure/recovery, message loss, and asynchronous concurrent execution.
While WPaxos helps most for slashing WAN latencies, it is also suitable for intra-datacenter deployments for its high-availability and throughput benefits.

\section{Related Work}
\label{sec:related}
Here we give an overview of core Paxos protocol and provide an architectural classification of algorithms in the Paxos family.

\subsection{Paxos Protocol}

\begin{figure}[t]
	\vspace*{-2mm}
	\centering
	\includegraphics[width=\columnwidth]{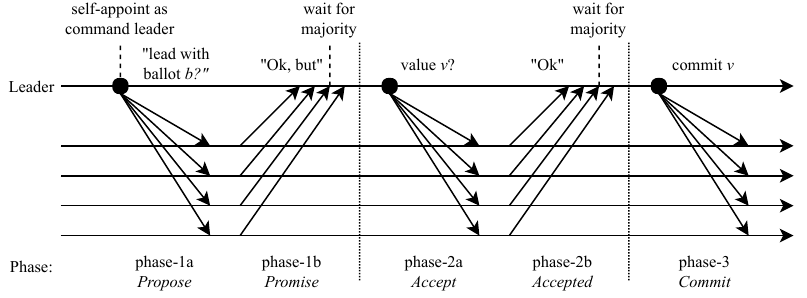}
	\caption{Overview of Paxos algorithm.}
	\label{fig:paxos}
	\vspace*{-2mm}
\end{figure}

Paxos separates its operation into three phases as illustrate in Figure~\ref{fig:paxos}.
%
In phase-1 a node proposes itself as a leader over some ballot $b$. Other nodes accept the proposal only if ballot $b$ is the highest they have seen so far. If some node has seen a leader with a greater ballot, that node will reject ballot $b$. Receiving a rejection fails the aspiring leader and causes it to start again with higher ballot. However, if the majority of the nodes accepts the leader, it will move to the phase-2 of the protocol. In phase-1, leader also learns uncommitted commands from earlier ballots to finish them later. 

In phase-2, the leader tells its followers to accept a command into their log. The command depends on the results of the prior phase, as the leader is obliged to finish highest ballot uncommitted command it may have learned earlier. Similarly, to phase-1, a leader requires a majority ack to complete phase-2. However, if a follower has learned of a higher ballot leader, it will not accept the command, and reject the leader, causing the leader to go back to leader election phase and retry. Once a majority of nodes ack to accept the command, the command becomes anchored and cannot be lost even in case of failures or leader changes, since that command is guaranteed to be learned by any future leader. 

Finally, Paxos commits the command in phase-3. In this phase, the leader sends a message to all followers to commit the command in their respective logs.

Many practical Paxos systems typically continue with the same leader for many rounds to avoid paying the cost of phase-1 repeatedly. This optimization, commonly known as Multi-Paxos~\cite{paxosmadecomplex}, iterates over phase-2 for different slots of the same ballot number. The safety is preserved as a new leader needs to obtain a majority of followers with a higher ballot, and this causes the original leader to get rejected in phase-2 and stops its progress.

\subsection{Paxos Variants}
Many Paxos variants optimize Paxos for specific needs, significantly extending the original protocol. We categorize non-byzantine consensus into five different classes as illustrated in Figure~\ref{fig:protocol_architectures}, and provide an overview of these protocol families.

\begin{figure}[!t]
\vspace*{-2mm}
\centering
\begin{subfigure}{0.23\textwidth}
\includegraphics[width=\linewidth]{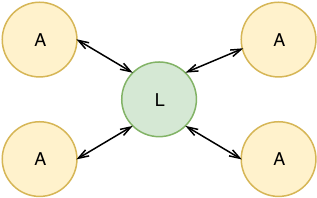}
\caption{Single-leader}
\label{fig:pax_single}
\end{subfigure}
\hspace*{\fill} 
\begin{subfigure}{0.23\textwidth}
\includegraphics[width=\linewidth]{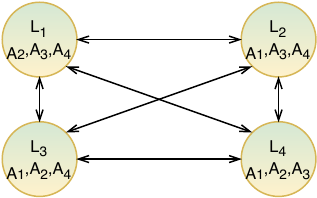}
\caption{Multi-leader}
\label{fig:pax_multi}
\end{subfigure}

\begin{subfigure}{0.23\textwidth}
\includegraphics[width=\linewidth]{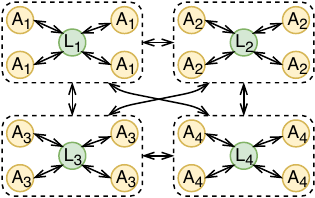}
\caption{Multi-leader multi-quorum}
\label{fig:pax_multiq}
\end{subfigure}
\hspace*{\fill} 
\begin{subfigure}{0.23\textwidth}
\includegraphics[width=\linewidth]{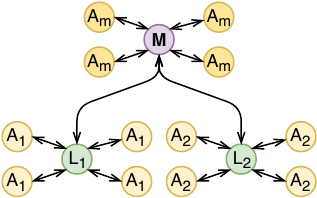}
\caption{Hierarchical}
\label{fig:pax_hierarchichal}
\end{subfigure}
\hspace*{\fill} 
\begin{subfigure}{0.23\textwidth}
\includegraphics[width=\linewidth]{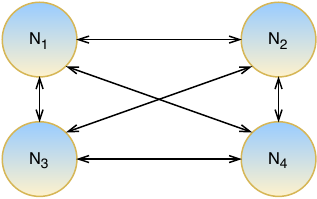}
\caption{Leaderless}
\label{fig:pax_leaderless}
\end{subfigure}
\hspace*{\fill} 
\caption{Overview of Paxos-based consensus protocol architectures. L designates leader nodes, and A is for acceptor nodes}
\label{fig:protocol_architectures}
\vspace*{-2mm}
\end{figure}

Single leader protocols, such as Multi-Paxos and Raft~\cite{raft} rely on one node to drive the progress forward (Figure~\ref{fig:pax_single}). Due to the overheads imposed by communications with all of the followers, the leader node often becomes a bottleneck in single leader protocols. Mencius~\cite{mencius} tries to reduce load imbalance by rotating the single leader in a round-robin fashion. 

Multi-leader protocols in Figure~\ref{fig:pax_multi}, such as $M^2$Paxos~\cite{m2paxos}, ZooNet~\cite{zoonet} operate on the observation that not all commands require to have a total order in the system. Multi-leader algorithms can run many commands in parallel at different leaders, as long as these commands belong to different conflict domains. This results in each leader node often serving as an acceptor to leaders of other conflict domains. Unlike single leader architectures that guarantee a single total order of commands in the system, multi-leader protocols provide a partial total order, where only commands belonging to the same conflict domain are ordered with respect to each other. The node despecialization allows multi-leader consensus to improve resource utilization in the cluster by spreading the load more evenly.

WPaxos goes one step further and allows different leaders to use different quorums, depicted in Figure~\ref{fig:pax_multiq}, as long as inter-quorum communication can ensure required safety properties. Such multi-leader, multi-quorum setup helps with both WAN latency and throughput due to smaller and geographically localized quorums.
The DPaxos data-management/replication protocol~\cite{dpaxos} cites our original WPaxos technical report~\cite{wpaxosRep} and adopts a similar protocol for the edge computing domain to bring highly granular high access locality data to the consumers at the edge.

Hierarchical multi-leader protocols, such as WanKeeper~\cite{wankeeper} and Vertical Paxos~\cite{vpaxos}, establish a chain of command between the leaders or quorums. A higher-level leader oversees and coordinates the lower-level children leaders/quorums. In a two-layer WanKeeper, master leader is responsible for assigning conflict domain ownership to lower-level leaders. Additionally, the master leader also handles operations that are of high demand by many lower-level quorums as to avoid changing the leadership back-and-forth. In Vertical Paxos, the master cluster is responsible for overseeing the lower-level configurations and does not participate in handling actual commands from the clients. Compared to flat multi-quorum setup of the WPaxos, hierarchical composition has quorum specialization, since master quorum is responsible for different or additional work compared to its children quorums.

Leaderless solutions in Figure~\ref{fig:pax_leaderless} also build on the idea of parallelizing the execution of non-conflicting commands. Unlike multi-leader approaches, however, leaderless systems, such as EPaxos~\cite{epaxos} do not impose the partitioning of conflict domains between nodes, and instead try to opportunistically commit any command at any node. Any node in EPaxos becomes an opportunistic leader for a command and tries to commit it by running a phase-2 of Paxos in a fast quorum system. If some other node in the fast quorum is also working on a conflicting command, then an additional round of communication is used to establish order on the conflicting commands.

\section{WPaxos Overview}
\label{sec:protocol}



We assume a set of \emph{nodes} communicating through message passing in an asynchronous environment. The nodes are deployed in a set of \emph{zones}, which are the unit of availability isolation. Depending on the deployment, a zone can range from a cluster or datacenter to geographically isolated regions. Each node is identified by a tuple consisting of a zone ID and node ID, i.e. $Nodes \triangleq 1..Z \times 1..N$.

Every node maintains a sequence of instances ordered by an increasing \emph{slot} number. Every instance is committed with a \emph{ballot} number. Each ballot has a unique leader. Similar to Paxos implementation~\cite{paxosmadecomplex}, we construct the ballot number as lexicographically ordered pairs of an integer and its leader identifier, s.t. $Ballots \triangleq Nat \times Nodes$. Consequently, ballot numbers are unique and totally ordered, and any node can easily retrieve the id of the leader from a given ballot.


\subsection{WPaxos Quorums}
\label{sec:quorum}

WPaxos leverages on the flexible quorums idea \cite{fpaxos}. This result shows that we can weaken Paxos' ``all quorums should intersect'' assertion to instead ``only quorums from different phases should intersect''. That is, majority quorums are not necessary for Paxos, provided that phase-1 quorums ($Q_1$) intersect with phase-2 quorums ($Q_2$). Flexible-Paxos, i.e., FPaxos, allows trading off $Q_1$ and $Q_2$ sizes to improve performance. Assuming failures and resulting leader changes are rare, phase-2 (where the leader tells the acceptors to decide values) is run more often than phase-1 (where a new leader is elected). Thus it is possible to improve performance of Paxos by reducing the size of $Q_2$ at the expense of making the infrequently used $Q_1$ larger.

\begin{definition}
\label{def:quorum}
A {quorum system} over the set of nodes is \emph{safe} if the quorums used in phase-1 and phase-2, named $Q_1$ and $Q_2$, intersect. That is, $\forall q_1 \in Q_1, q_2 \in Q_2 : q_1 \cap q_2 \neq \emptyset$.
\end{definition}


WPaxos adopts the flexible quorum idea to WAN deployments. Our quorum system derives from the grid quorum layout, shown in Figure \ref{fig:grid}, in which rows and columns act as $Q_1$ and $Q_2$ quorums respectively. An attractive property of this grid quorum arrangement is $Q_1+Q_2$ does not need to exceed $N$, the total number of acceptors, in order to guarantee intersection of any $Q_1$ and $Q_2$. Let $q_1, q_2$ denote one specific instance in $Q_1$ and $Q_2$. Since $q_1 \in Q_1$ are chosen from rows and $q_2 \in Q_2$ are chosen from columns, any $q_1$ and $q_2$ are guaranteed to intersect even when $|q_1+q_2| < N$.

\begin{figure}[!t]
\centering
	\begin{subfigure}{0.23\textwidth}
	\includegraphics[width=\linewidth]{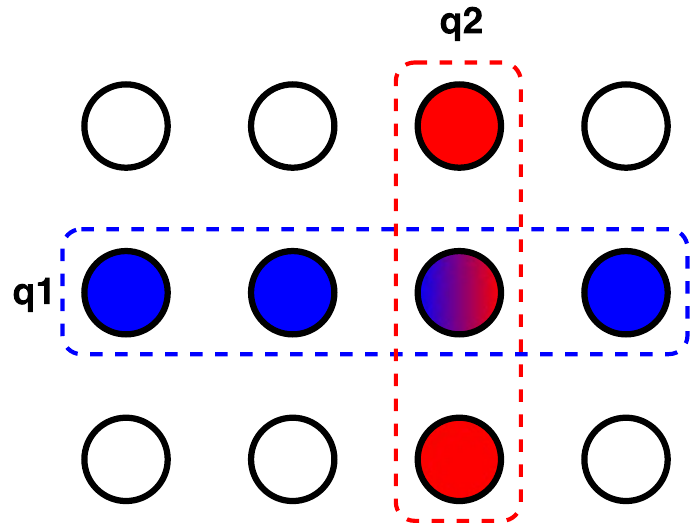}
	\caption{}
	\label{fig:grid}
	\end{subfigure}
\hspace*{\fill} 
	\begin{subfigure}{0.23\textwidth}
	\includegraphics[width=\linewidth]{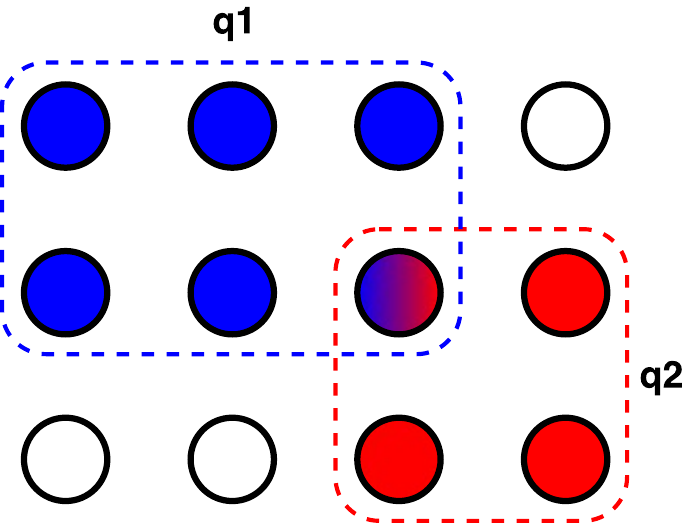}
	\caption{}
	\label{fig:gridb}
	\end{subfigure}
\caption{(a) 4-by-3 grid with a $Q_1$ quorum in rows and a $Q_2$ quorum in columns when $f_n=f_z=0$ (b) 4-by-3 grid with a $Q1$ and $Q2$ quorum when $f_n=f_z=1$.}
\label{fig:quorums}
\end{figure}

In WPaxos quorums, each column represents a zone and acts as a unit of availability or geographical partitioning. The collection of all zones form a grid. In this setup, we further generalize the grid quorum constraints in both $Q_1$ and $Q_2$ to achieve a more fault-tolerant and flexible alternative. Instead of using rigid grid columns, we introduce two parameters: $f_z$, the number of zone failures tolerated, and $f_n$, the number of node failures a zone can tolerate before losing availability.

In order to tolerate $f_n$ crash failures in every zone, WPaxos picks $f_n+1$ nodes in a zone over $l$ nodes, regardless of their row position. In addition, to tolerate $f_z$ \emph{zone failures} within $Z$ zones, $q_1 \in Q_1$ is selected from $Z-f_z$ zones, and $q_2 \in Q_2$ from $f_z+1$ zones.
Below, we formally define WPaxos quorums in TLA+~\cite{tla} and prove that the $Q1$ and $Q2$ quorums always intersect.
{\footnotesize

\begin{align*}
	Q_1 \triangleq \{ & q \in \textbf{SUBSET } \text{Nodes} : \footnotemark \\
                          &\text{Cardinality}(q) = (f_n+1) \times (Z - f_z) \wedge \\
                          &\neg \exists k \in  \textbf{SUBSET } q : \forall i, j \in k : i[1] = j[1] \wedge \text{Cardinality}(k) > {f_n\!+\!1} \}
\end{align*}
\begin{align*}
	Q_2 \triangleq \{ & q \in \textbf{SUBSET } \text{Nodes} : \\
                          & \text{Cardinality}(q) = (l-f_n) \times (f_z + 1) \wedge \\
                            & \neg \exists k \in  \textbf{SUBSET } q : \forall i, j \in k : i[1] = j[1] \wedge \text{Cardinality}(k) > {l\!-\!f_n} \}
\end{align*}
}
\footnotetext{\textbf{SUBSET } S is the set of subsets of S}

\begin{lemma}
	WPaxos $Q_1$ and $Q_2$ quorums satisfy intersection requirement (Definition \ref{def:quorum}).
\end{lemma}

\begin{proof}
	(1) WPaxos $q_1$s involve $Z-f_z$ zones and $q_2$s involve $f_z+1$ zones, since $Z-f_z+f_z+1 = Z+1 > Z$, there is at least one zone selected by both quorums.
	(2) Within the common zone, $q_1$ selects $f_n+1$ nodes and $q_2$ selects $l-f_n$ nodes out of $l$ nodes forming a zone. Since $l-f_n+f_n+1 > l$, there is at least one node in the intersection.
\end{proof}

Figure~\ref{fig:grid} shows a 4-by-3 grid with $f_n=f_z=0$ and
Figure~\ref{fig:gridb} shows a 4-by-3 with $f_n=f_z=1$. In the latter deployment, each zone has 3 nodes, and each $q_2$ includes 2 out of 3 nodes from 2 zones. The $q_1$ quorum spans 3 out of 4 zones and includes any 2 nodes from each zone. Using a 2 row $q_1$ rather than 1 row $q_1$ has negligible effect on the performance (as we show in Section~\ref{sec:eval}) and provides more fault-tolerance.


When using grid-based flexible quorums (as opposed to unified majority quorums), the total number of faults tolerated, $F$, becomes topology-dependent. In EPaxos quorums are selected from a set so it is not important which nodes you pick. In WPaxos quorums are selected from a grid, so the location of the nodes picked becomes important.

We define $F_{min}$ to be the size of F in the worst possible case of failure placement so as to violate availability with respect to read/write operations on any item. If all the faults conspire to target a particular $q \in Q2$, then after Cardinality($Q2$) faults, the $q \in Q2$ is wiped off. By definition any $Q2$ quorum intersects with any $Q1$ quorum. By wiping off a $Q2$ quorum in its entirety, the faults made any $Q1$ quorum unavailable as well. To account for the case where Cardinality of $Q1$ quorum could be less than that of $Q2$, we make the formula symmetric and define as follows. 

{\footnotesize
$F_{min} = Min (\text{Cardinality}(Q2), \text{Cardinality}(Q1)) -1$
}

We define $F_{max}$ to be the size of $F$ in the best possible case of failure placement in the grid. In this case, the faults miss a union of a $Q1$ quorum and $Q2$ quorum, leaving at least one $Q1$ and $Q2$ quorum intact so WPaxos can continue operating. Note that the $Q2$ quorum may be completely embedded inside the $Q1$ quorum (or vice versa if $Q1$ quorums are smaller than $Q2$ quorums). So the formula is derived as subtracting from N, the cardinality of $Q1$ and $Q2$, and by adding the maximum cardinality of intersection of $Q1$ and $Q2$.

{\footnotesize
$F_{max} = N\!-\!\text{Cardinality}(Q1)\!-\!\text{Cardinality}(Q2)\!+\! (f_z\!+\!1)\!*\!(f_n\!+\!1)$ 
}

For the 4-by-3 grid with $f_n=f_z=0$ in Figure~\ref{fig:grid}, $F_{min}=2$ and $F_{max}=6$. For the deployment in Figure~\ref{fig:gridb} with $f_n=f_z=1$, $F_{min}=3$, and $F_{max}=6$.
\footnote{For a deployment of size 2F+1, fast-quorum is of size F+(F+1)/2. Therefore for N=12 and F=5 EPaxos fast quorum is 8, and EPaxos can tolerate upto 4 failures before the fast quorum size is breached. After 4 failures, EPaxos operations detoriate because they time out waiting response from the nonexisting fast quorum, and then proceed to go to the second round to be able to make progress with majority nodes. Progress is still possible until 5 node failures. The EPaxos paper suggests a reconfiguration to be invoked upon node failures to reduce N, and shrink the fast quorum size.}

\subsection{Multi-leader}
\label{sec:multi-leader}

In contrast to FPaxos~\cite{fpaxos} which uses flexible quorums with a classical single-leader Paxos protocol, WPaxos presents a multi-leader protocol over flexible quorums. Every node in WPaxos can act as a leader for a subset of objects in the system. This allows the protocol to process requests for objects under different leaders concurrently. Each object in the system is allotted its own commit log, allowing for per-object linearizability. A node can lead multiple objects at once, all of which may have different ballot and slot numbers in their corresponding logs.

The WPaxos protocol consists of two phases. The concurrent leaders \emph{steal} ownership/leadership of objects from each other using phase-1 of Paxos executed over $q_1 \in Q_1$. Then phase-2 commits the update-requests to the object over $q_2 \in Q_2$, selected from the leader's zone (and nearby zones) for improved locality. The leader  can execute phase-2 multiple times until some other node steals the object.

The phase-1 of the protocol starts only when the node needs to steal an object from a remote leader or if a client has a request for a brand new object that is not in the system. This phase of the algorithm causes the ballot number to grow for the object involved.
After a node becomes the owner/leader for an object, it repeats phase-2 multiple times on that object for committing commands/updates, incrementing the slot number at each iteration, while the ballot number for the object stays the same.

Figure~\ref{fig:algo} shows the normal operation of both phases, and also references each operation to the algorithms in Section~\ref{sec:algorithm}.


\begin{figure}[t]
\centering
\includegraphics[width=\columnwidth]{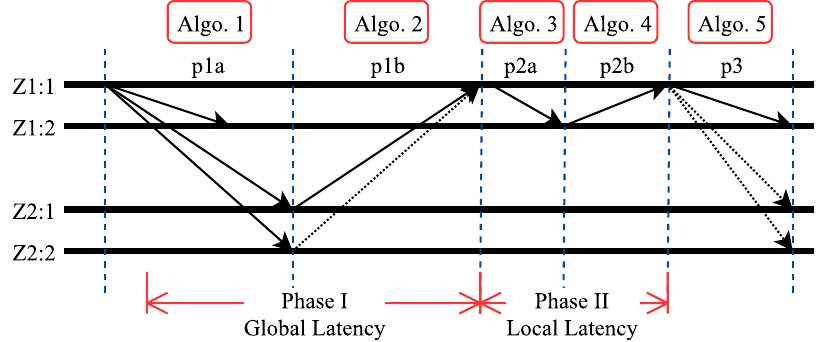}
\caption{Normal case messaging flow}
\label{fig:algo}
\end{figure}

\subsection{Object Stealing}
\label{sec:steal}

When a node needs to steal an object from another leader in order to carry out a client request, it first consults its internal cache to determine the last ballot number used for the object and performs phase-1 on some $q_1 \in Q_1$ with a larger ballot. Object stealing is successful if the candidate node can out-ballot the existing leader. This is achieved in just one phase-1 attempt, provided that the local cache is current and a remote leader is not engaged in another phase-1 on the same object.

Once the object is stolen, the old leader cannot act on it, since the object is now associated with a higher ballot number than the ballot it had at the old leader. This is true even when the old leader was not in the $q_1$ when the key was stolen, because the intersected node in $q_2$ will reject any object operations attempted with the old ballot. The object stealing  may occur when some commands for the objects are still in progress, therefore, a new leader must recover any accepted, but not yet committed commands for the object.

WPaxos maintains separate ballot numbers for all objects isolating the effects of object stealing. Keeping per-leader ballot numbers, i.e., keeping a single ballot number for all objects maintained by the leader, would necessitate out-balloting all objects of a remote leader when trying to steal one object. This would then create a leader dueling problem in which two nodes try to steal different objects from each other by constantly proposing a higher ballot than the opponent.
Using separate ballot numbers for each object alleviates ballot contention, although it can still happen when two leaders are trying to take over the same object currently owned by a third leader. To mitigate that issue, we use two additional safeguards: (1) resolving ballot conflict by zone ID and node ID in case the ballot counters are the same, and (2) implementing a random back-off mechanism in case a new dueling iteration starts anyway. 

Object stealing is part of core WPaxos protocol. In contrast to the simplicity and agility of object stealing in WPaxos, object relocation in other systems require integration of another service, such as movedir in Spanner~\cite{spanner}, or performing multiple reconfiguration or coordination steps as in Vertical Paxos~\cite{verticalpaxos}.
Vertical Paxos depends on a reliable master service that overseeing configuration changes. Object relocation involves configuration change in the node responsible for processing commands on that object. When a node in a different region attempts to steal the object, it must first contact the reconfiguration master to obtain the current ballot number and next ballot to be used. The new leader then must complete phase-1 of Paxos on the old configuration to learn the previous commands. Upon finishing the phase-1, the new leader can commit any uncommitted slots with its own set of acceptors. At the same time the new leader notifies the master of completing phase-1 with its ballot. Only after the master replies and activates the new configuration, the leader can start serving user requests. This process can be extended to multiple objects, by keeping track of separate ballot numbers for each object. Vertical Paxos requires three separate WAN communications to change the leadership, while WPaxos can do so with just one WAN communication.

\section{WPaxos Algorithm}
\label{sec:algorithm}


In the basic algorithm, every node maintains a set of variables and a sequence of commands written into the command \emph{log}. The command log can be committed out of order, but has to be executed against the state machine in the same order without any gap. Every command accesses only one object $o$. Every node leads its own set of objects in a set called $own$.


\begin{algorithm}[ht]
\caption*{\textbf{process}(self $\in$ Nodes) Initialization}
\footnotesize
\begin{algorithmic}[1]
\Statex \textbf{variables}
\State $ballots = [o \in Objects \mapsto \langle 0, self \rangle];$ 
\State $slots = [o \in Objects \mapsto 0];$ 
\State $own = \{\}$ 
\State $
	\begin{aligned}[t]
		log = &[o \in Objects \mapsto \\
		&[s \in Slots \mapsto \\
		&[b \mapsto 0, v \mapsto \langle\rangle, c \mapsto \text{FALSE}]]];
	\end{aligned}$
\end{algorithmic}
\end{algorithm}

All nodes in WPaxos initialize their state with above variables. We assume no prior knowledge of the ownership of the objects; a user can optionally provide initial object assignments. The highest known ballot numbers for objects are constructed by concatenating counter=0 and the node ID (line 1). The slot numbers start from zero (line 2), and the objects self owned is an empty set (line 3). Inside the log, an instance contains three components, the ballot number $b$ for that slot, the proposed command/value $v$ and a flag $c$ indicates whether the instance is committed (line 4).

\subsection{Phase-1: Prepare}

\begin{algorithm}[ht]
\caption{Phase-1a}
\label{alg:p1a}
\footnotesize
\begin{algorithmic}[1]
\State \textbf{macro} p1a () \{
	\State \indent \textbf{with} ($o \in Objects$) \{
	\State \indent\indent \textbf{await} $o \notin own$;
	\State \indent\indent $ballots[o] := \langle ballots[o][1]+1, self\rangle$;
	\State \indent\indent 
	$\openup-1.5\jot
	\begin{aligned}[t]
		Send([type &\mapsto \textbf{``1a''}, \\
		n &\mapsto self, \\
		o &\mapsto o, \\
		b &\mapsto ballots[o]]); \} \}
	\end{aligned}$
\end{algorithmic}
\end{algorithm}

WPaxos starts with a client sending requests to one of the nodes. A client typically chooses a node in the local zone to minimize the initial communication costs. The request message includes a command and some object $o$ on which the command needs to be executed. Upon receiving the request, the node checks if the object exists in the set of $own$, and start phase-1 for any new objects by invoking \textbf{p1a()} procedure in Algorithm~\ref{alg:p1a}. If the object is already owned by this node, the node can directly start phase-2 of the protocol.
In p1a(), a larger ballot number is selected and \textbf{``1a''} message is sent to a $Q_1$ quorum.

\begin{algorithm}[ht]
\caption{Phase-1b}
\label{alg:p1b}
\footnotesize
\begin{algorithmic}[1]
\State \textbf{macro} p1b () \{
	\State \indent \textbf{with} ($m \in msgs$) \{
		\State \indent\indent \textbf{await} $m.type = \textbf{``1a''}$;
		\State \indent\indent \textbf{await} $m.b \succeq ballots[m.o]$;
		\State \indent\indent $ballots[m.o] := m.b$;
		\State \indent\indent \textbf{if} ($o \in own$) $own := own \setminus \{m.o\}$;
		\State \indent\indent 
		$\openup-1.5\jot
		\begin{aligned}[t]
			Send([type &\mapsto \textbf{``1b''}, \\
			n &\mapsto self, \\
			o &\mapsto m.o, \\
			b &\mapsto m.b, \\
			s &\mapsto slots[m.o]]); \} \}
		\end{aligned}$
\end{algorithmic}
\end{algorithm}

The \textbf{p1b()} procedure processes the incoming \textbf{``1a''} message sent during phase-1 initiation. A node can accept the sender as the leader for object $o$ only if the sender's ballot $m.b$ is greater or equal to the ballot number it knows of (line 4). If object $o$ is owned by current node, it is removed from set $own$ (line 6). Finally, the \textbf{``1b''} message acknowledging the accepted ballot number is send (line 7). The highest slot associated with $o$ is also attached to the reply message, so that any unresolved commands can be committed by the new leader.

\subsection{Phase-2: Accept}

Phase-2 of the protocol starts after the completion of phase-1 or when it is determined that no phase-1 is required for a given object. WPaxos carries out this phase on a $Q_2$ quorum residing in the closest $F+1$ zones, thus all communication is kept local, greatly reducing the latency.

\begin{algorithm}[ht]
\caption{Phase-2a}
\label{alg:p2a}
\footnotesize
\begin{algorithmic}[1]
	\State $\openup-1.5\jot
	\begin{aligned}[t]
		Q_1\text{Satisfied}(o, b) \triangleq & \exists q \in Q_1 : \forall n \in q : \exists m \in msgs : \\
		& \land m.type = \textbf{``1b''} \\
		& \land m.o = o \\
		& \land m.b = b \\
		& \land m.n = n
	\end{aligned}$
\Statex
\State \textbf{macro} p2a () \{
	\State \indent \textbf{with} ($m \in msgs$) \{
		\State \indent\indent \textbf{await} $m.type = \textbf{``1b''}$;
		\State \indent\indent \textbf{await} $m.b = \langle ballots[m.o][1], self\rangle$;
		\State \indent\indent \textbf{await} $m.o \notin own$;
		\State \indent\indent \textbf{if} ($Q_1$Satisfied($m.o, m.b$)) \{
		\State \indent\indent\indent $own := own \cup \{m.o\}$;
		\State \indent\indent\indent $slots[m.o] := slots[m.o] + 1$;
		\State \indent\indent\indent $\openup-1.5\jot
		\begin{aligned}[t]
			log[m.o][slots[m.o]] := [b &\mapsto m.b, \\
			v &\mapsto \langle slots[m.o], self \rangle, \\
			c &\mapsto \text{FALSE}];
		\end{aligned}$
		\State \indent\indent\indent $\openup-1.5\jot
		\begin{aligned}[t]
			Send([type &\mapsto \textbf{``2a''}, \\
			n &\mapsto self, \\
			o &\mapsto m.o, \\
			b &\mapsto m.b, \\
			s &\mapsto slots[m.o], \\
			v &\mapsto \langle slots[m.o], self \rangle]); \} \} \}
		\end{aligned}$
\end{algorithmic}
\end{algorithm}

Procedure \textbf{p2a()} in Algorithm~\ref{alg:p2a} collects the ``1b'' messages for itself (lines 4-6). The node becomes the leader of the object only if $Q_1$ quorum is satisfied (line 7,8). The new leader then recovers any uncommitted slots with suggested values and starts the accept phase for the pending requests that have accumulated in queue. Phase-2 is launched by increasing the highest slot (line 9), and creates new entry in log (line 10), sending \textbf{``2a''} message (line 11).

\begin{algorithm}[ht]
\caption{Phase-2b}
\label{alg:p2b}
\footnotesize
\begin{algorithmic}[1]
\State \textbf{macro} p2b () \{
	\State \indent \textbf{with} ($m \in msgs$) \{
		\State \indent\indent \textbf{await} $m.type = \textbf{``2a''}$;
		\State \indent\indent \textbf{await} $m.b \succeq ballots[m.o]$;
		\State \indent\indent $ballots[m.o] := m.b$;
		\State \indent\indent $log[m.o][m.s] := [b \mapsto m.b, v \mapsto m.v, c \mapsto \text{FALSE}]$;
		\State \indent\indent $\openup-1.5\jot
		\begin{aligned}[t]
			Send([type &\mapsto \textbf{``2b''}, \\
			n &\mapsto self, \\
			o &\mapsto m.o, \\
			b &\mapsto m.b, \\
			s &\mapsto m.s]); \} \}
		\end{aligned}$
\end{algorithmic}
\end{algorithm}

Once the leader of the object sends out the ``2a'' message at the beginning of phase-2, the replicas respond to this message as shown in Algorithm~\ref{alg:p2b}. The leader node updates its instance at slot $m.s$ only if the message ballot $m.b$ is greater or equal to accepted ballot (line 4-6).

\subsection{Phase-3: Commit}

The leader collects replies from its $Q_2$ acceptors. The request proposal either gets committed with replies satisfying a $Q_2$ quorum, or aborted if some acceptors reject the proposal citing a higher ballot number. In case of rejection, the node updates a local ballot and puts the request in this instance back to main request queue to retry later.

\begin{algorithm}[h]
\caption{Phase-3}
\label{alg:p3}
\footnotesize
\begin{algorithmic}[1]
\State $\openup-1.5\jot
	\begin{aligned}[t]
		Q_2\text{Satisfied}(o, b, s) \triangleq & \exists q \in Q_2 : \forall n \in q : \exists m \in msgs : \\
		& \land m.type = \textbf{``2b''} \\
		& \land m.o = o \\
		& \land m.b = b \\
		& \land m.s = s \\
		& \land m.n = n
	\end{aligned}$
\Statex
\State \textbf{macro} p3 () \{
	\State \indent \textbf{with} ($m \in msgs$) \{
		\State \indent\indent \textbf{await} $m.type = \textbf{``2b''}$;
		\State \indent\indent \textbf{await} $m.b = \langle ballots[m.o][1], self \rangle$;
		\State \indent\indent \textbf{await} $log[m.o][m.s].c \neq \text{TRUE}$;
		\State \indent\indent \textbf{if} ($Q_2$Satisfied($m.o, m.b, m.s$)) \{
			\State \indent\indent\indent $log[m.o][m.s].c := \text{TRUE}$;
			\State \indent\indent\indent $\openup-1.5\jot
			\begin{aligned}[t]
				Send([type &\mapsto \textbf{``3''}, \\
				n &\mapsto self, \\
				o &\mapsto m.o, \\
				b &\mapsto m.b, \\
				s &\mapsto m.s, \\
				v &\mapsto log[m.o][m.s].v]); \} \} \}
			\end{aligned}$
\end{algorithmic}
\end{algorithm}

\subsection{Properties}
\label{sec:properties}


\textbf{Non-triviality.} For any node $n$, the set of committed commands is always a sequence $\sigma$ of proposed commands, i.e. $\exists \sigma : committed[n] = \bot \bullet \sigma$. Non-triviality is straightforward since nodes only start phase-1 or phase-2 for commands proposed by clients in Algorithm 1.

\textbf{Stability.} For any node $n$, the sequence of committed commands at any time is a prefix of the sequence at any later time, i.e. $\exists \sigma : committed[n] = \gamma$ {\em at} $t \implies committed[n] = \gamma\bullet\sigma$ at $t+\Delta$. Stability asserts any committed command cannot be overridden later. It is guaranteed and proven by Paxos that any leader with higher ballot number will learn previous values before proposing new slots. WPaxos inherits the same process.

\textbf{Consistency.} For any slot of any object, no two leaders can commit different values. This property asserts that object stealing and failure recovery procedures do not override any previously accepted or committed values. We verified this consistency property by model checking a TLA+ specification of WPaxos algorithm. 

WPaxos consistency guarantees are on par with other protocols, such as EPaxos, that solve the generalized consensus problem~\cite{generalizedconsensus}. Generalized consensus relaxes the consensus requirement by allowing non-interfering commands to be processed concurrently. Generalized consensus no longer enforces a totally ordered set of commands. Instead only conflicting commands need to be ordered with respect to each other, making the command log a partially ordered set. WPaxos maintains separate logs for every object and provides per-object linearizability. 

\textbf{Liveness.} A proposed command $\gamma$ will eventually be committed by all non-faulty nodes $n$, i.e. $\diamond \forall n \in \text{Nodes}: \gamma \in committed[n].$ The PlusCal code presented in Section~\ref{sec:algorithm} specifies what actions each node is allowed to perform, but not when to perform, which affects liveness. The liveness property satisfied by WPaxos algorithm is same to that of ordinary Paxos: as long as there exists $q_1 \in Q_1$ and $q_2 \in Q_2$ are alive, the system will progress.

\section{Extensions}
\label{sec:extensions}

\subsection{Locality Adaptive Object Stealing}
\label{sec:adaptive}

\begin{figure}[t]
\centering
\begin{subfigure}{0.48\columnwidth}
\includegraphics[width=\linewidth]{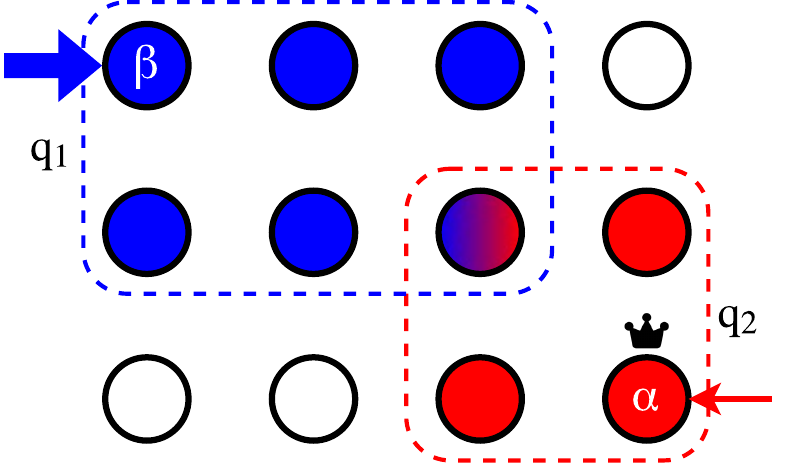}
\caption{}
\end{subfigure}
\hspace*{\fill} 
\begin{subfigure}{0.48\columnwidth}
\includegraphics[width=\linewidth]{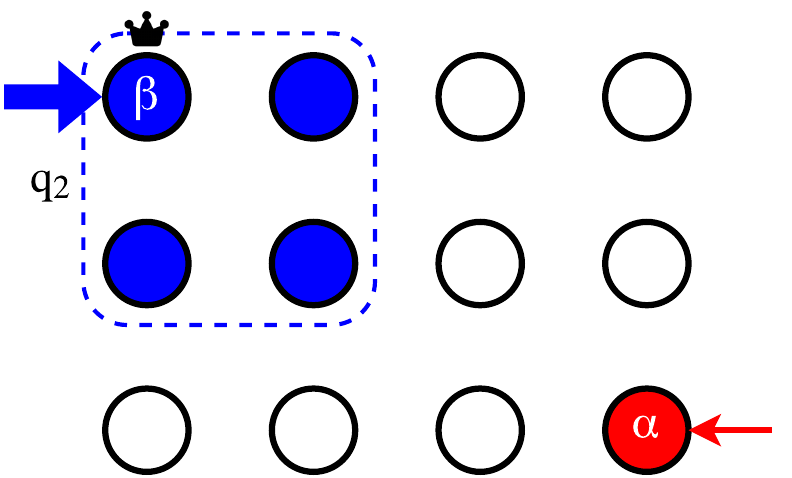}
\caption{}
\end{subfigure}
\caption{(a) Initial leader $\alpha$ observes heavy cross-region traffic from node $\beta$, thus triggers $\beta$ to start phase-1 on its $q_1$.  (b) $\beta$ becomes new leader and benefits more on the workload.}
\label{fig:leader_election}
\end{figure}


The basic protocol migrates the object from a remote region to a local region upon the first request, but that causes a performance degradation when an object is frequently accessed across many zones.
With locality adaptive object stealing we can delay or deny the object transfer to a zone issuing the request based on an object migration policy. The intuition behind this approach is to move objects to a zone whose clients will benefit the most from not having to communicate over WAN, while allowing clients accessing the object from less frequent zones to get their requests forwarded to the remote leader.

Our \textit{majority-zone} migration policy aims to improve the locality of reference by transferring the objects to zones that sending out the highest number of requests for the objects, as shown in Figure~\ref{fig:leader_election}. Since the current object leader handles all the requests, it has the information about which clients access the object more frequently. If the leader $\alpha$ detects that the object has more requests coming from a remote zone, it will initiate the object handover by communicating with the node $\beta$, and in its turn $\beta$ will start the phase-1 protocol to steal the leadership of that object.

\subsection{Replication Set}
\label{sec:rq2}

WPaxos provides flexibility in selecting a replication set. The phase-2 (p2a) message need not be broadcast to the entire system, but only to a subset of $Q_2$ quorums, denoted as a replication $Q_2$ or $RQ_2$. The user has the freedom to choose the replication factor across zones from the minimal required $F+1$ zones up to the total number of $Z$ zones. Such choice can be seen as a trade off between communication overhead and a more predictable latency, since the replication zone may not always be the fastest to reply. Additionally, if a node outside of the $RQ_2$ becomes the new leader of the object, that may delay the new phase-2 as the leader need to catch up with the missing logs in previous ballots. One way to minimize the delay is let the $RQ_2$ reply on phase-2 messages for replication, while the potential leader nodes learn the states as non-voting learners.

\begin{figure}[t]
\vspace{5mm}
\centering
\includegraphics[width=0.8\columnwidth]{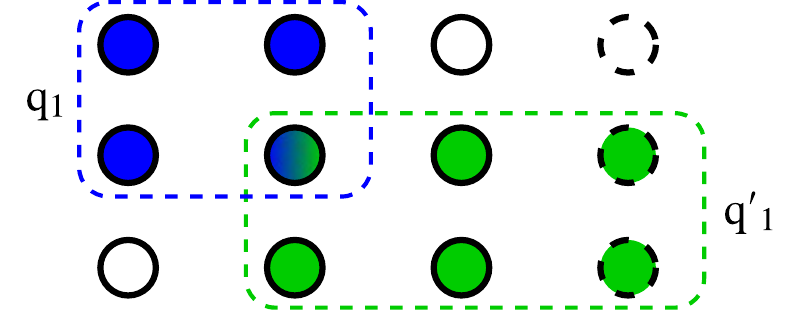}
\vspace{2mm}
\caption{Adding new zone (dashed)}
\label{fig:reconfig}
\vspace{5mm}
\end{figure}

\subsection{Fault Tolerance and Reconfiguration}
\label{sec:fault-tolerance}

WPaxos can make progress as long as it can form valid $q_1$ and $q_2$ quorums. The flexibility of WPaxos enables the user to deploy the system with quorum configuration tailored to their needs. Some configurations are geared towards performance, while others may prioritize fault tolerance.
By default, WPaxos configures the quorums to tolerate one zone failure and minority node failures per zone, and thus provides similar fault tolerance as Spanner with Paxos groups deployed over three zones.

WPaxos remains partially available when more zones fail than the tolerance threshold it was configured for. In such a case, no valid $q_1$ quorum may be formed, which halts the object stealing routine, however the operations can proceed for objects owned in the remaining live regions, as long as there are enough zones left to form a $q_2$ quorum. 

\label{sec:reconfig}


The ability to reconfigure, i.e., dynamically change the membership of the system, is critical to provide reliability for long periods as it allows  crashed nodes to be replaced.
WPaxos achieves high throughput by allowing \emph{pipelining} (like Paxos and Raft algorithms) in which new commands may begin phase-2 before any previous instances/slots have been committed. Pipelining architecture brings more complexity to reconfiguration, as there may be another reconfiguration operation in the pipeline which could change the quorum and invalidate a previous proposal. Paxos~\cite{paxosmadecomplex} solves this by limiting the length of the pipeline window to $\alpha > 0$ and only activating the new config $C'$ chosen at slot $i$ until slot $i+\alpha$. Depending on the value of $\alpha$, this approach either limits throughput or latency of the system. On the other hand, Raft \cite{raft} does not impose any limitation of concurrency and proposes two solutions. The first solution is to restrict the reconfiguration operation, i.e. what can be reconfigured. For example, if each operation only adds one node or removes one node, a sequence of these operations can be scheduled to achieve arbitrary changes. The second solution is to change configuration in two phases: a union of both old and new configuration $C+C'$ is proposed in the log first, and committed by the quorums combined. Only after the commit, the leader may propose the new config $C'$. During the two phases, any election or command proposal should be committed by quorum in both $C$ and $C'$.
To ensure safety during reconfiguration, all these solutions essentially prevent two configurations $C$ and $C'$ to make decision at the same time that leads to divergent system states.

WPaxos adopts the more general two-phase reconfiguration procedure from Raft for arbitrary $C'$s, where $C = \langle Q_1, Q_2 \rangle$, $C' = \langle Q_1', Q_2' \rangle$. WPaxos further reduces the two phases into one in certain special cases since adding and removing one zone or one row operations are the most common reconfigurations in the WAN topology. These four operations are equivalent to the Raft's first solution because the combined quorum of $C+C'$ is equivalent to quorum in $C'$. We show one example of adding new zone of dashed nodes in the Figure~\ref{fig:reconfig}.

Previous configuration $Q_1$ involves two zones, whereas the new config $Q_1'$ involves three zones including the new zone added. The quorums in $Q_1'$ combines quorums in $Q_1$ is same as $Q_1'$. Both $Q_2$ and $Q_2'$ remains the same size of two zones.
The general quorum intersection assumption and the restrictions $Q_1' \cup Q_1 = Q_1'$ and $Q_2' \cup Q_2 = Q_2'$ ensure that old and new configuration cannot make separate decisions and provides same safety property.

\section{Evaluation}
\label{sec:eval}

We developed a general framework, called \textit{Paxi} to conduct our evaluation. The framework allows us to compare WPaxos, EPaxos, M2Paxos and other Paxos protocols in the same controlled environment under identical workloads. We implemented Paxi along with WPaxos and EPaxos in Go and released it as an open-source project on GitHub at \url{https://github.com/ailidani/paxi}.
The framework provides extended abstractions to be shared between all Paxos variants, including location-aware configuration, network communication, client library with RESTful API, and a quorum management module (which accommodates majority quorum, fast quorum, grid quorum and flexible quorum). Paxi's networking layer encapsulates a message passing model and exposes basic interfaces for a variety of message exchange patterns, and transparently supports TCP, UDP and simulated connection with Go channels. Additionally, our Paxi framework incorporates mechanisms to facilitate the startup of the system by sharing the initial parameters through the configuration management tool.


\subsection{Setup}
\label{sec:setup}

\begin{figure}
\centering
\includegraphics[width=0.8\columnwidth]{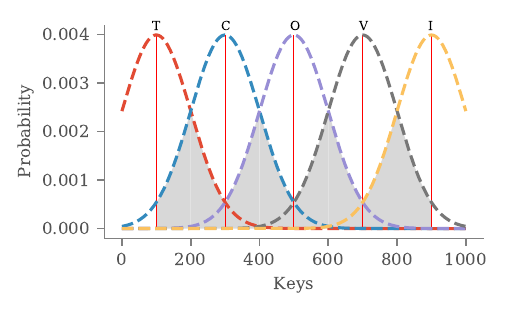}
\caption{Locality Workload}
\label{fig:loc}
\end{figure}

We evaluated WPaxos using the key-value store abstraction provided by our Paxi framework. We used AWS EC2 \cite{ec2} nodes to deploy WPaxos across 5 different regions: Tokyo (\textbf{T}), California (\textbf{C}), Ohio (\textbf{O}), Virginia (\textbf{V}), and Ireland (\textbf{I}).  In our experiments, we used 4 m5.large instances at each AWS region to host 3 WPaxos nodes and 20 concurrent clients. WPaxos in our experiments uses adaptive mode by default, unless otherwise noted.

We conducted all experiments with Paxi microbenchmark. Paxi provides similar benchmarking capabilities as YCSB~\cite{ycsb}, with both benchmarks generating similar workloads. However, Paxi benchmark adds more tuning knobs to facilitate testing in wide area with workloads that exhibit different conflict and locality characteristics.

In order to simulate workloads with tunable access locality patterns we used a normal distribution to control the probability of generating a request on each object. As shown in the Figure~\ref{fig:loc}, we used a pool of 1000 common objects, with the probability function of each region denoting how likely an object is to be selected at a particular zone. Each region has a set of objects it is more likely to access. We define \textbf{locality} as the percentage of the requests pulled from such set of likely objects.
We introduce locality to our evaluation by drawing the conflicting keys from a Normal distribution $\mathcal{N}(\mu, \sigma^2)$, where $\mu$ can be varied for different zones to control the locality, and $\sigma$ is shared between zones. The locality can be visualized as the non-overlapping area under the probability density functions in Figure~\ref{fig:loc}.

\begin{definition}
\textit{\textbf{Locality}} $L$ is the complement of the overlapping coefficient (OVL)\footnote{The overlapping coefficient (OVL) is a measurement of similarity between two probability distributions, refers to the shadowed area under two probability density functions simultaneously \cite{ovl}.} among workload distributions: $L = 1 - \widehat{OVL}$.
\end{definition}

Let $\Phi(\frac{x-\mu}{\sigma})$ denote the cumulative distribution function (CDF) of any normal distribution with mean $\mu$ and deviation $\sigma$, and $\hat{x}$ as the x-coordinate of the point intersected by two distributions, locality is given by $L = \Phi_1(\hat{x}) - \Phi_2(\hat{x})$. At the two ends of the spectrum, locality equals to 0 if two overlapping distributions are congruent, and locality equals to 1 if two distributions do not intersect.

\subsection{Object Space}

\begin{figure}[t]
\centering
\includegraphics[width=0.9\columnwidth]{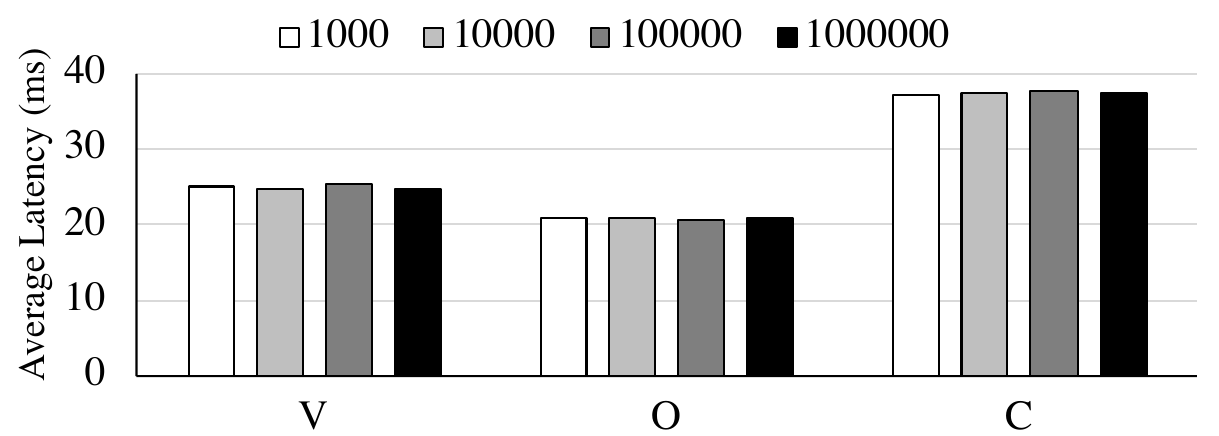}
\caption{Average latency for uniformly random workload with increasing number of objects.}
\label{fig:wp_keys}
\end{figure}

We begin by presenting our evaluation of the overhead with increasing number of objects in WPaxos system. Every object in WPaxos is fully replicated. We preload the system with one thousand to one million keys evenly distributed among three regions (Virginia, Oregon and California), then generate requests with random key from every region. To evaluate the performance impact, we measure the average latency in each one of the three regions.

The results shown in Figure \ref{fig:wp_keys} indicates there are no significant impacts on request latency. This is expected since a hash map index has O(1) lookup time to keep track of object and its current leader. The index data does not consume extra memory because the leader ID is already maintained in the ballot number from last Paxos log entry. At the end of our experiment, one million keys without log snapshots and garbage collection consumes about 1.6 GB memory out of our 8 GB VM. For more keys inserted into the system, we expect steady performance as long as they fit into the memory.

\begin{figure}[t]
\centering
\includegraphics[width=\columnwidth]{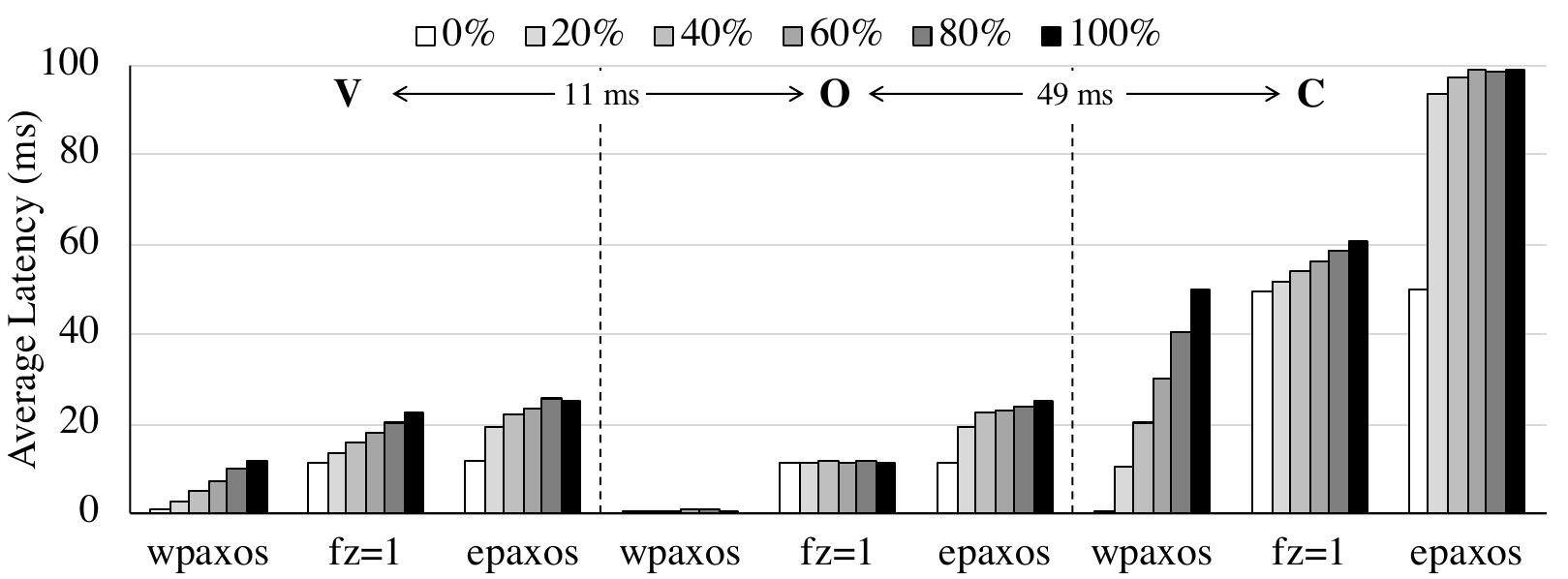}
\caption{Average latency with increasing conflict ratio for WPaxos ($f_z=0$ and $f_z=1$) and EPaxos}
\label{fig:wp_latency}
\end{figure}

\begin{figure}[h]
\centering
\includegraphics[width=0.9\columnwidth]{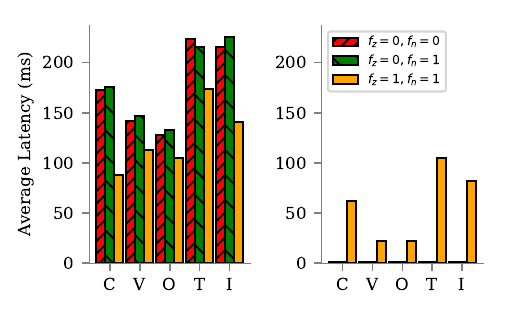}
\caption{Average latency for phase-1 (left) and phase-2 (right) in different quorum systems.}
\label{fig:quorum}
\end{figure}

\subsection{WPaxos Quorum Latencies}

In this set of experiments, we compare the latency of $Q_1$ and $Q_2$ accesses in three different fault tolerance configurations: ($f_z\!=\!f_n\!=\!0$), ($f_z\!=\!0; f_n\!=\!1$), and ($f_z\!=\!f_n\!=\!1$). The configurations with $f_n\!=\!0$ uses a single node per zone/region for $Q_1$, requiring all nodes in one zone/region to form $Q_2$, while configurations with $f_n\!=\!1$ require one fewer node in $Q_2$. When $f_z\!=\!0$, $Q_1$ uses all 5 zones and $Q_2$ remains in a single region. With $f_z\!=\!1$, $Q_1$ uses 4 zones which reduce phase-1 latency significantly, but $Q_2$ requires 2 zones/regions thus exhibits WAN latency. In each region we simultaneously generated a fixed number (1000) of phase-1 and phase-2 requests, and measured the latency for each phase. Figure~\ref{fig:quorum} shows the average latency in phase-1 (left) and phase-2 (right).

Quorum size of $Q_1$ in $f_n\!=\!1$ configurations is half of that for WPaxos with $f_n\!=\!0$, but both experience average latency of about one round trip to the farthest peer region, since the communication happens in parallel. Within a zone, however, $f_n\!=\!1$ can tolerate one straggler node, reducing the latency for the most frequently used $Q_2$ quorum type.



\subsection{Conflicting Commands}

Here we evaluate the performance of WPaxos in terms of conflicting commands.  In WPaxos, we treat any two commands operating on the same object in the Paxi key-value store as conflicting. In our experiments, the workload ranges from 0\% conflicts (i.e. all requests are completely local to its leader) to 100\% conflicts (i.e. every request targets the conflicting object). 

While WPaxos can only denote object-based conflicts and non-conflicts, EPaxos can denote operation-based conflicts and non-conflicts in the general case. In the context of our experiments on the Paxi key-value store, for EPaxos, we treat any two update operations on the same object as conflicting, and treat read operations on any object as nonconflicting with any other read operations. This is the same setup used in the evaluation of the EPaxos paper~\cite{epaxos}. 

As shown in Figure \ref{fig:wp_latency}, WPaxos without zone-failure-tolerance ($f_z\!=\!0$) performance better than WPaxos that tolerate one zone failure ($f_z\!=\!1$) in every case, because $Q_2$ within a region avoids the RTT between neighboring regions for non-conflicting commands. Since the Ohio region is located in the relative center of our topology, it becomes the leader of conflicting objects and the performance in that region becomes independent of conflicts. More interestingly, even though both WPaxos $f_z\!=\!1$ and EPaxos requires involving two RTTs for conflicting commands, WPaxos is able to reduce the latency by committing requests with two closer regions. For example, in 100\% conflict workload, requests from C is committed by one RTT between C and O (49ms) plus one RTT between V and O (11ms) instead of two RTTs of CO like EPaxos.


\begin{figure*}[t]
\centering
\begin{subfigure}{0.32\textwidth}
	\includegraphics[width=\linewidth]{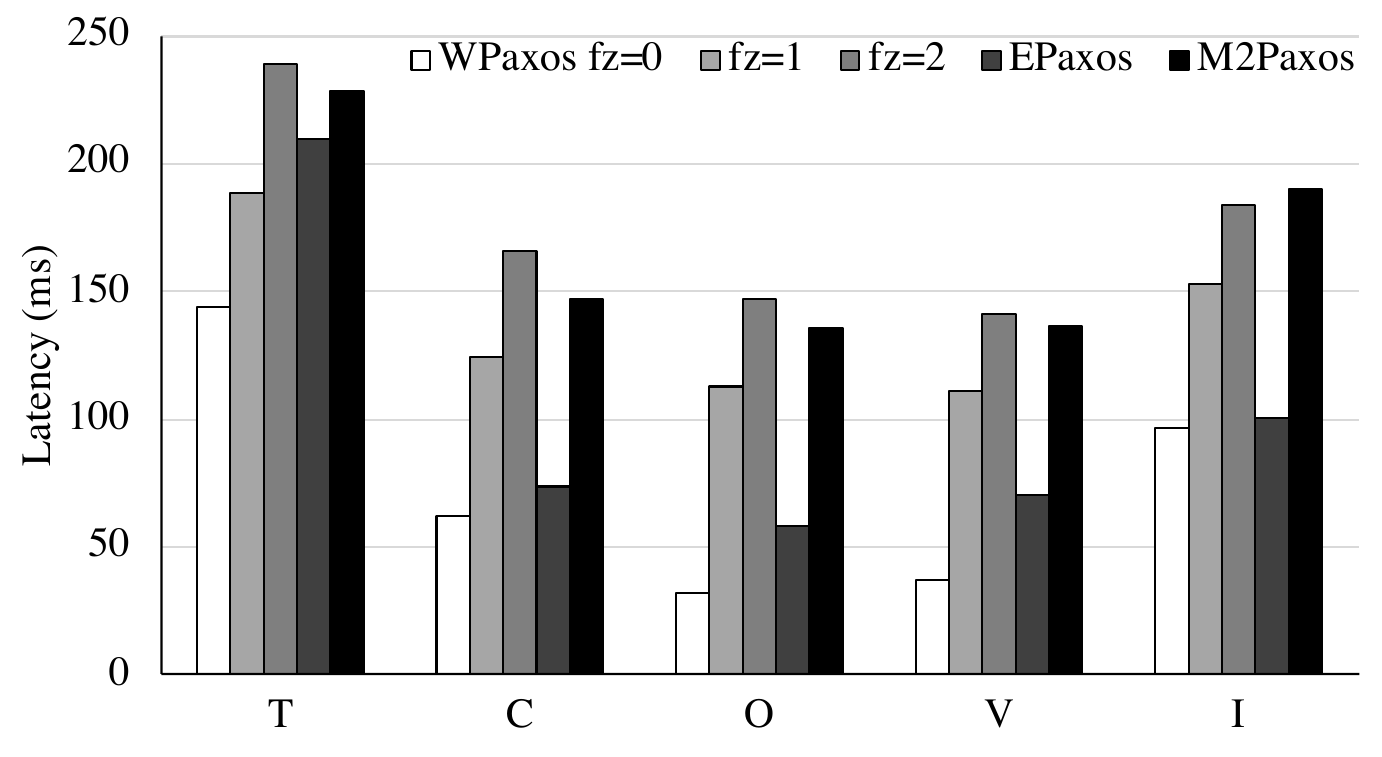}
	\caption{Uniformly random workload}
	\label{fig:random_lat}
\end{subfigure}
\hspace*{\fill} 
\begin{subfigure}{0.32\textwidth}
	\includegraphics[width=\linewidth]{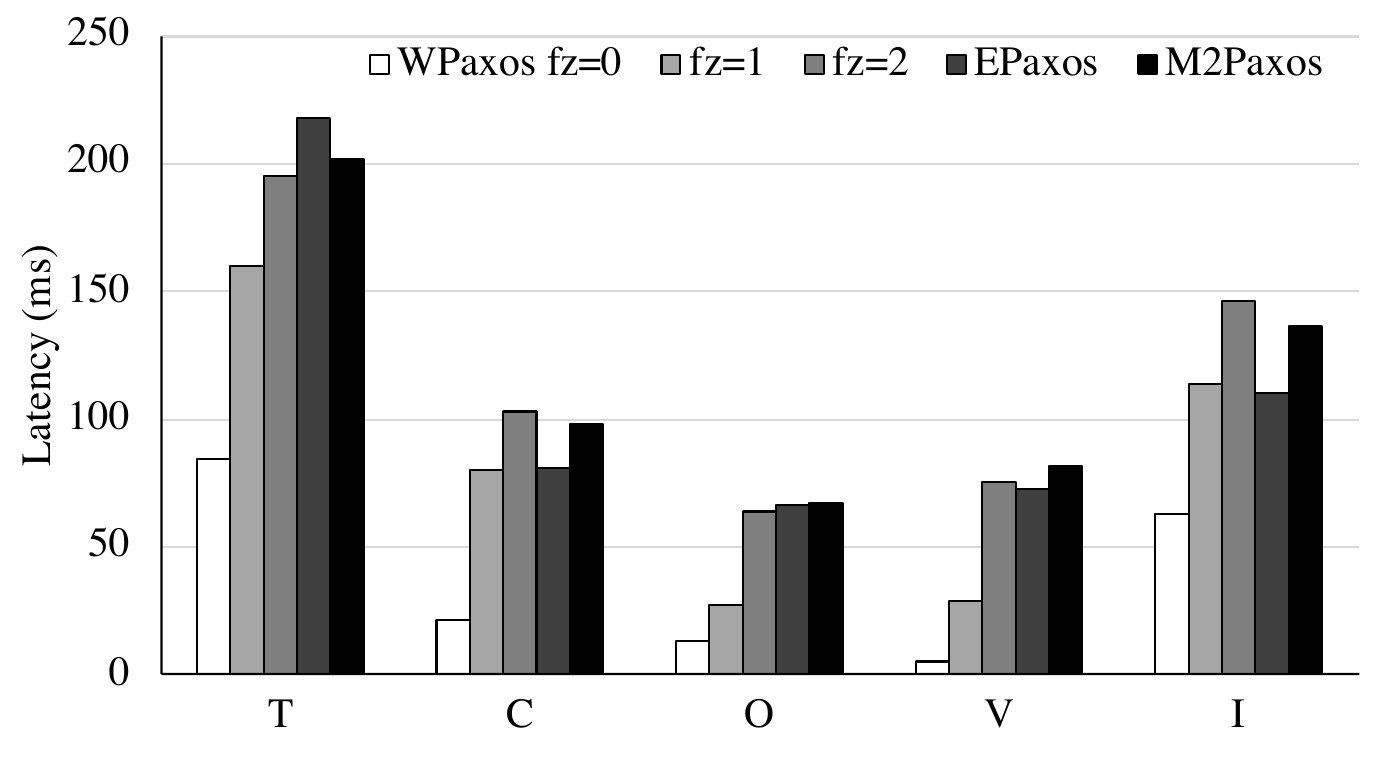}
	\caption{Medium Locality ($\sigma=100$) workload}
	\label{fig:locality_lat70}
\end{subfigure}
\hspace*{\fill} 
\begin{subfigure}{0.32\textwidth}
	\includegraphics[width=\linewidth]{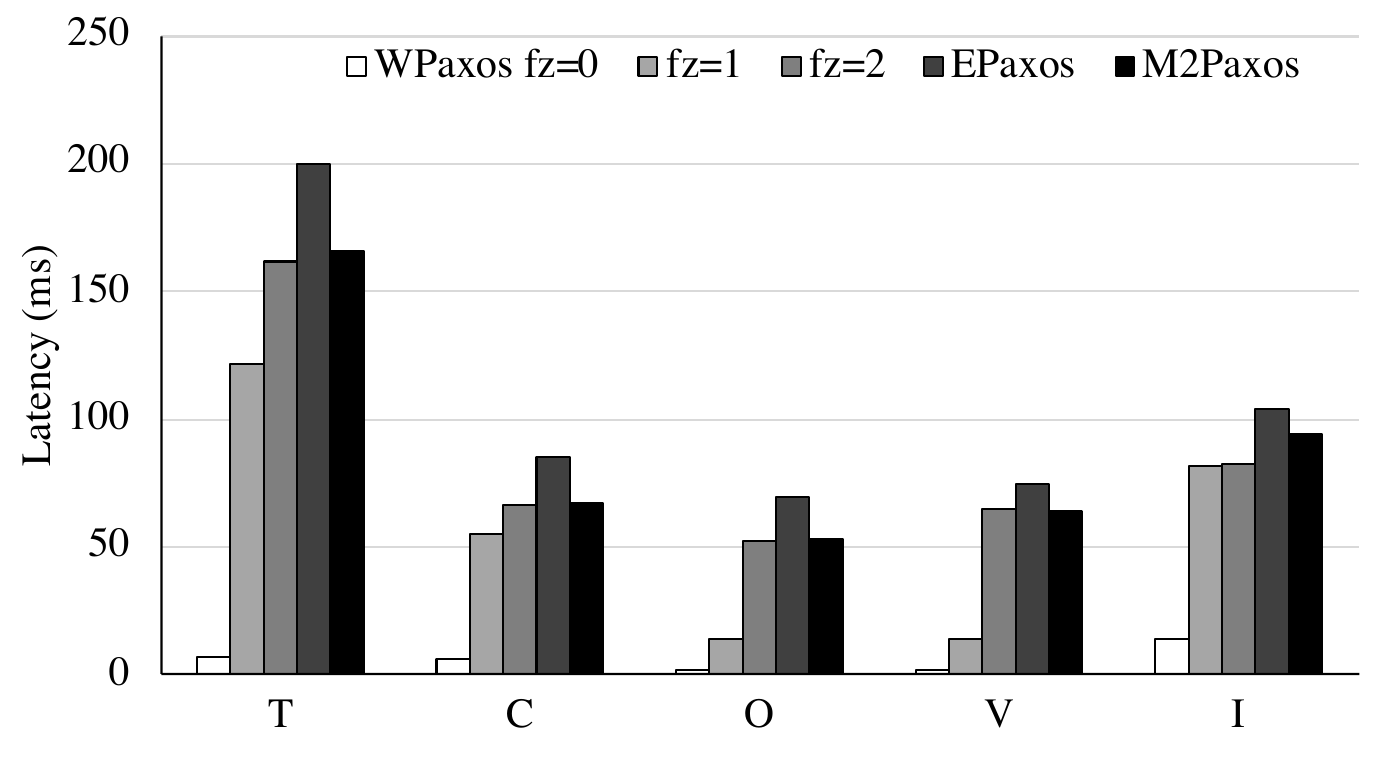}
	\caption{High Locality ($\sigma=50$) workload}
	\label{fig:locality_lat90}
\end{subfigure}
\caption{Average latencies in different regions.}
\end{figure*}

\begin{figure}[t]
\centering
\includegraphics[width=0.8\linewidth]{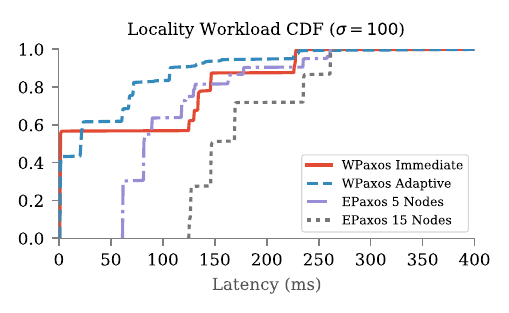}
\caption{CDF of 70\% locality workload WPaxos $f_z=0$ immediate/adaptive mode and EPaxos}
\label{fig:cdf}
\end{figure}

\subsection{Latency Comparison}

We compare the latency of WPaxos (with $f_z=0,1,2$), EPaxos, and M2Paxos protocols using three sets of workloads: random (Figure \ref{fig:random_lat}), $\sim$70\% locality (Figure \ref{fig:locality_lat70}), and $\sim$95\% locality (Figure \ref{fig:locality_lat90}). Before each round of experiment, we divide 1000 objects and preload them into each region, such that  every region owns 200 objects according to Figure \ref{fig:loc}. For each experiment, the clients in each region generate requests concurrently within the duration of one minute. 
 
Figure~\ref{fig:random_lat} compares the average latency of random workload in 5 regions. Each region experiences different latencies due to their asymmetrical location in the geographical topology. WPaxos in $f_z\!=\!1$ and $f_z\!=\!2$ tolerance configurations show higher latency than EPaxos because requests are forwarded to leaders in other regions most of the time and causes extra wide area RTT, whereas EPaxos initiates PreAccept phase by any local leader. WPaxos with $f_z=0$ performs better than all other protocols due to its local phase 2 quorums.


Figure \ref{fig:locality_lat70} shows that, under $\sim$70\% locality workload ($\mathcal{N}(\mu_z, \sigma=100)$), regions located close to the geographic center improve their average latencies. Given the wide standard deviation of accessing keys, EPaxos experiences slightly higher conflict rate, and WPaxos and M2Paxos still experience request forwarding to remote leaders. When the phase 2 quorums cover the same number of regions, all three protocols show a similar average latency. Since WPaxos provides more flexibility in configuring the fault tolerance factor, WPaxos $f_z=0, 1$ outperforms all other protocols in all regions.



In Figure \ref{fig:locality_lat90}, we increase the locality to $\sim$95\% ($\mathcal{N}(\mu_z, \sigma=50)$). EPaxos shows similar pattern as previous experiments, whereas WPaxos achieves much lower latency by avoiding WAN forwarding largely in all regions.

Figure \ref{fig:cdf} shows the tail latencies caused by object stealing in WPaxos immediate and adaptive modes and compares them with EPaxos with 5 and 15 node deployments. In the figure, all request latencies from every region are aggregated to produce the cumulative distribution (CDF). Using WPaxos immediate, the edge regions suffer from high object stealing latencies because their Q1 latencies are longer due to their location in the topology. WPaxos adaptive alleviates and smoothens these effects. Even under low locality, about half of the requests are committed in local-area latency in WPaxos.

\subsection{Throughput Comparison}

We experiment on scalability of WPaxos with respect to the number of requests it processes by driving a steady workload at each zone. Instead of the {\em medium} instances, we used a cluster of 15 {\em large} EC2 nodes to host WPaxos deployments. EPaxos is hosted at the same nodes, but with only one EPaxos node per zone. We opted out of using EPaxos with 15 nodes, because our preliminary experiments showed significantly higher latencies with such a large EPaxos deployment.
We limit WPaxos deployments to a single leader per zone to be better comparable to EPaxos. We gradually increase the load on the systems by issuing more requests and measure the  latency at each of the throughput levels. Figure~\ref{fig:tp_vs_lat} shows the latencies as the aggregate throughput increases.

At low load, we observe both immediate and adaptive WPaxos significantly outperform EPaxos as expected. With relatively small number of requests coming through the system, WPaxos has low contention for object stealing, and can perform many operations locally within a region. As the number of requests increases and contention rises, performance of both EPaxos and WPaxos with immediate object stealing deteriorates. EPaxos suffers from high conflict in WAN degrading its performance further. This is because cross-datacenter latencies increase commit time, resulting in higher conflict probability. At high conflict, Paxos goes into two-phase operation, which negatively impacts both latency and maximum throughput.

\begin{figure}[t]
	\centering
	\begin{subfigure}{0.49\columnwidth}
	\includegraphics[width=\linewidth]{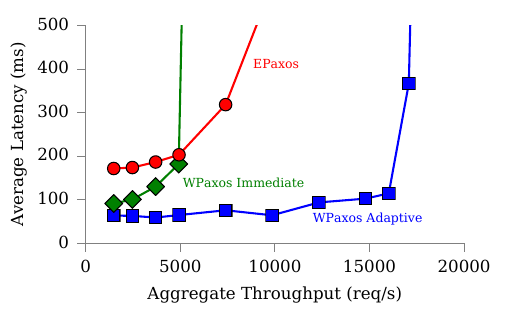}
	\end{subfigure}
	\begin{subfigure}{0.49\columnwidth}
	\includegraphics[width=\linewidth]{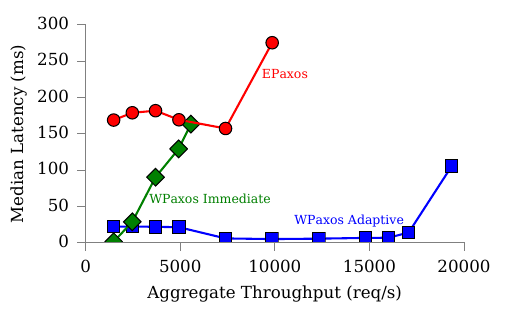}
	\end{subfigure}
	\caption{Request latency as the throughput increases.}
	\label{fig:tp_vs_lat}
\end{figure}


Immediate WPaxos suffers from leaders competing for objects with neighboring
regions, degrading its performance faster than EPaxos.
Median request latency graph in Figure~\ref{fig:tp_vs_lat} clearly illustrates this deterioration.
This behavior in WPaxos with immediate object
stealing is caused by dueling leaders: as two nodes in neighboring zones try to
acquire ownership of the same object, each restarts phase-1 of the protocol
before the other leader has a chance to finish its phase-2. 

On the other hand, WPaxos in adaptive object stealing mode scales better and shows almost no degradation until it starts to reach the CPU and networking limits of individual instances. Adaptive WPaxos median latency actually slightly decreases under the medium workloads, while EPaxos shows gradual latency increase. At the workload of 10000 req/s adaptive WPaxos outperforms EPaxos 9 times in terms of average latency and 54 times in terms of median latency. 


\subsection{Shifting Locality Workload}

Many applications in the WAN setting may experience workloads with shifting access patterns such as diurnal patterns~\cite{adapt_data_migration, gmach2007workload}.
Figure~\ref{fig:shifting} illustrates the effects of shifting locality in the workload on WPaxos and statically key-partitioned Paxos (KPaxos). KPaxos starts in the optimal state with most of the requests done on the local objects. 
When the access locality is gradually shifted by changing the mean of the locality distributions at a rate of 2 objects/sec, the access pattern shifts further from optimal for statically partitioned Paxos, and its latency increases. 
WPaxos, on the other hand, does not suffer from the shifts in the locality. The adaptive algorithm slowly migrates the objects to regions with more demand, providing stable and predictable performance under shifting access locality.

\begin{figure}[t]
\centering
\includegraphics[width=0.7\columnwidth]{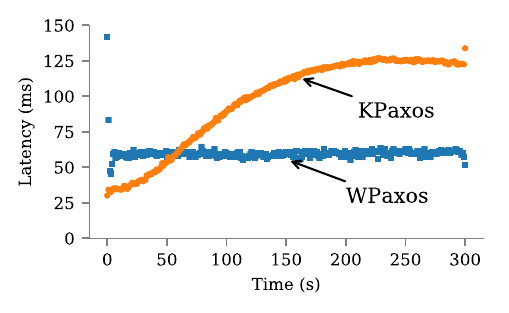}
\caption{The average latency in each second.}
\label{fig:shifting}
\end{figure}


\subsection{Fault Tolerance}

\begin{figure}[t]
\centering
\begin{subfigure}{\columnwidth}
    \includegraphics[width=\linewidth]{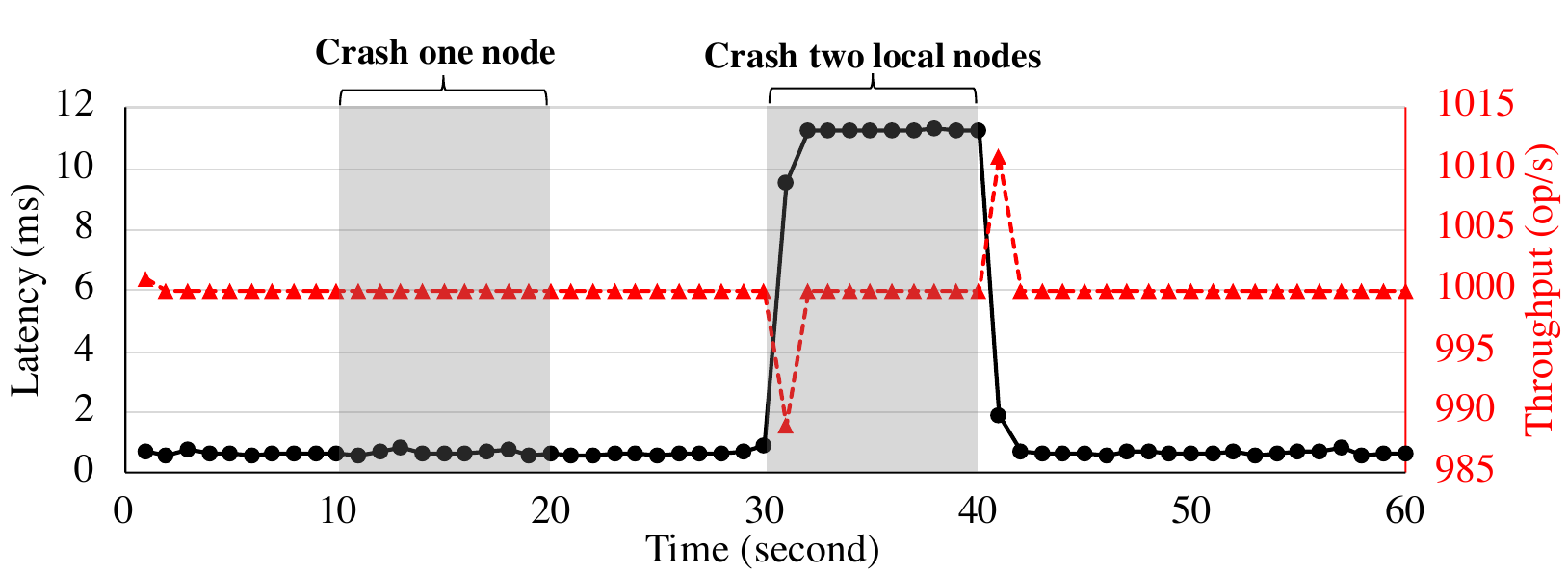}
    \caption{WPaxos ($f_z=0$) crash one and two local nodes in a zone}
    \label{fig:wp_crash_f0}
\end{subfigure}
\begin{subfigure}{\columnwidth}
    \includegraphics[width=\linewidth]{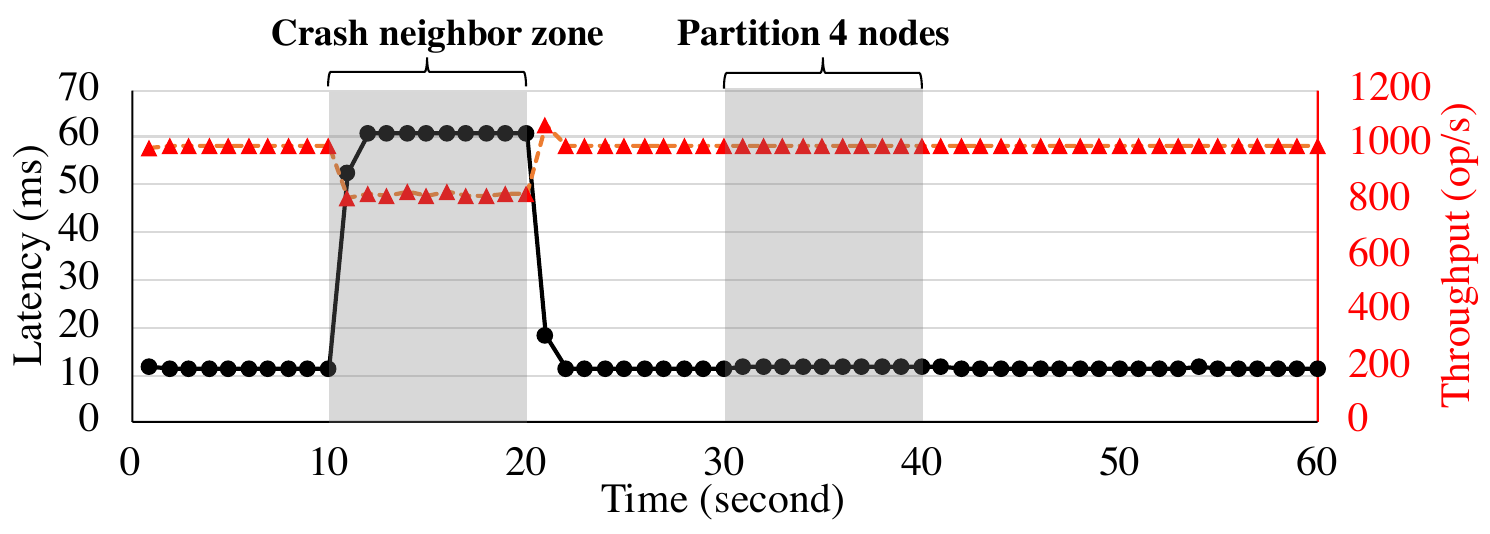}
    \caption{WPaxos ($f_z=1$) crash one zone and partition 4 nodes out of 9}
    \label{fig:wp_crash_f1}
\end{subfigure}
\caption{WPaxos availability}
\end{figure}

In this section we evaluate WPaxos availability by using Paxi framework fault injection API to introduce different failures and measure latency and throughput of normal workload in every second. Every fault injection will last for 10 seconds and recover.

Figure \ref{fig:wp_crash_f0} shows the result of first deployment where $f_n = 1$ and $f_z = 0$. The throughput and latency is measured in region V. For the first 10 seconds under normal operation, the latency and throughput is steady at less than 1 millisecond and 1000 operations/second respectively. We crash one node in region V first, it does not have any effect on performance since $|q_2|=2$ out of 3 nodes in that region. At 30th second, we crash two local nodes so that a local $q_2$ cannot be formed. The requests has to wait for two acks from neighboring region O, which introduce additional 11 ms RTT to the latency.

Figure \ref{fig:wp_crash_f1} shows the results of a same deployment but $f_z = 1$ where we can tolerate any one zone failure. The latency remains at 11 ms as $q_2$ requires 2 nodes from both V and O. Until 10th second, we crash region O entirely, The leader has to wait for acks from region C and latency become 60 ms. When region O recovers, we partitioned 4 nodes as the minority from the system of 9 nodes. The 4 nodes including 3 nodes from C and one node from O. As expected, such partition does not have any effect on system performance.

In all above failures, WPaxos always remain available.

\section{Concluding Remarks}
\label{sec:concl}
WPaxos achieves fast wide-area coordination by dynamically partitioning the objects across multiple leaders that are strategically deployed using flexible quorums. Such partitioning and emphasis on local operations allow our protocol to significantly outperform other WAN Paxos solutions.
Since the object stealing is an integrated part of phase-1 of Paxos, WPaxos remains simple as a pure Paxos flavor and obviates the need for another service/protocol for relocating objects to zones. Since the base WPaxos protocol guarantees safety to concurrency, asynchrony, and faults, the performance can be tuned orthogonally and aggressively. In future work, we will investigate smart object stealing policies that can proactively move objects to zones with high demand. We will also investigate implementing transactions more efficiently leveraging WPaxos optimizations.

\ifCLASSOPTIONcompsoc
  \section*{Acknowledgments}
\else
  \section*{Acknowledgment}
\fi
This project is in part sponsored by the National Science Foundation (NSF) under award number CNS-1527629.

\bibliographystyle{IEEEtran}
\bibliography{murat,ailidani,acharapk,tevfik}

\end{document}